\documentclass[letterpaper, 10pt, conference]{ieeeconf}  
\IEEEoverridecommandlockouts                              
\usepackage{mathrsfs}
\usepackage{epsfig}
\usepackage{amsmath,cite}
\usepackage{amsfonts, amssymb}
\usepackage{chemarrow}
\usepackage{graphicx}
\usepackage{hyperref}
\usepackage[T1]{fontenc}
\usepackage{stmaryrd}
\usepackage{dcolumn}
\usepackage{graphicx}
\usepackage{amssymb}
\usepackage{amsmath}
\usepackage{dcolumn}
\usepackage{bm}
\usepackage{color}
\usepackage{pdfsync}
\usepackage{epstopdf}
\usepackage{algorithm}
\usepackage{algorithmic}
\usepackage{nomencl}
\usepackage{multirow}
\usepackage{fancyhdr}
\usepackage{datetime}
\usepackage{graphicx}  
\usepackage{url}
\usepackage{hyperref}
\usepackage{subfigure}

\makenomenclature
\DeclareMathOperator*{\argmin}{argmin}

\newcommand{\prior}{\varphi}

\newcommand{\E}{\mathbb{E}}

\newcommand{\Noise}{\boldsymbol{\Pi}}
\newcommand{\Covinv}{\mathbf{S}}

\newcommand{\Cov}{\boldsymbol{\Pi}}

\newcommand{\n}{N}

\newcommand{\hyper}{\gamma}

\newcommand{\Hyper}{\boldsymbol{\gamma}}

\newcommand{\R}{\mathbb{R}}
\newcommand{\balpha}{\boldsymbol{\alpha}}

\newcommand{\bc}{\mathbf{c}}

\newcommand{\bff}{\mathbf{f}}

\newcommand{\bu}{\mathbf{u}}
\newcommand{\bv}{\mathbf{v}}
\newcommand{\bw}{\mathbf{w}}
\newcommand{\bx}{\mathbf{x}}

\newcommand{\by}{\mathbf{y}}
\newcommand{\bz}{\mathbf{z}}

\newcommand{\bI}{\mathbf{I}}

\newcommand{\Prob}{\mathcal{P}}

\newcommand{\bN}{\mathcal{N}}
\newcommand{\bG}{\boldsymbol{\Gamma}}

\newcommand{\bgamma}{\boldsymbol{\gamma}}

\newcommand{\bGamma}{\boldsymbol{\Gamma}}
\newcommand{\bxi}{\boldsymbol{\xi}}
\newcommand{\bXi}{\boldsymbol{\xi}}

\newcommand{\dic}{\mathbf{A}}

\newcommand{\dicc}{\boldsymbol{\Psi}}
\newcommand{\hyperprior}{\varphi}

\newcommand{\define}{\triangleq}

\newcommand{\diag}{\operatorname{diag}}

\newcommand{\blkdiag}{\operatorname{blkdiag}}
\newcommand{\trace}{\operatorname{Tr}}
\newcommand{\mean}{\mathbf{m}_{\bw}}
\newcommand{\variance}{\mathbf{\Sigma}_{\bw}}

\newcommand{\suponly}[1]{}

\newcommand{\yc}{\mathbf{y}^{\left[c\right]}}
\newcommand{\Ac}{\dic^{\left[c\right]}}
\newcommand{\wc}{\bw^{\left[c\right]}}
\newcommand{\noisec}{\bxi^{\left[c\right]}}

\newcommand{\ync}{\mathbf{y}_n^{\left[c\right]}}
\newcommand{\Anc}{\dic_n^{\left[c\right]}}
\newcommand{\wnc}{\bw_n^{\left[c\right]}}
\newcommand{\noisenc}{\bxi_n^{\left[c\right]}}

\newcommand{\MSE}{\left( \by-\dic\bw\right)^{\top}\Covinv \left( \by-\dic\bw\right)}
\newcommand{\MSEk}{\left( \by-\dic\bw\right)^{\top}\Covinv^{k} \left( \by-\dic\bw\right)}
\newlength{\bracewidth}
\newcommand{\myunderbrace}[2]{\settowidth{\bracewidth}{$#1$}#1\hspace*{-1\bracewidth}\smash{\underbrace{\makebox{\phantom{$#1$}}}_{#2}}}

\newtheorem{proposition}{Proposition}
\newtheorem{assumption}{Assumption}
\newtheorem{remark}{Remark}

\usepackage{multicol}  
\usepackage{hyperref}
\begin{document} 

\title{Identifying Biochemical Reaction Networks From Heterogeneous Datasets}  
\author{Wei Pan, Ye Yuan$^*$, Lennart Ljung, Jorge Gon\c{c}alves and Guy-Bart Stan
	\thanks{W.~Pan and G.-B.~Stan are with the Centre for Synthetic Biology and
		Innovation and the Department of Bioengineering, Imperial College London, United Kingdom. Email: {\tt\small  \{w.pan11, g.stan\}@imperial.ac.uk}.}
	\thanks{Ye Yuan was with the Control Group, Department of Engineering, University of Cambridge, United Kingdom. He is with Department of Electrical Engineering and Computer Sciences, UC Berkeley. J. Gon\c{c}alves is with the Control Group, Department of Engineering, University of Cambridge, United Kingdom and with the Luxembourg Centre for Systems Biomedicine, Luxembourg. Email: {\tt\small  jmg77@cam.ac.uk}. $^*$ For correspondence: {\tt\small  yy311@berkeley.edu}.} 
\thanks{L. Ljung is with Division of Automatic Control, Department of Electrical Engineering, Link\"{o}ping University, Sweden. Email: {\tt ljung@isy.liu.se}.}
\thanks{The authors gratefully acknowledge the support of Microsoft Research through the PhD Scholarship Program of Mr Wei Pan. Dr Ye Yuan acknowledges the support from EPSRC (project EP/I03210X/1). Dr Guy-Bart Stan gratefully acknowledges the support of the EPSRC Centre for Synthetic Biology and Innovation at Imperial College London (project EP/G036004/1) and of the EPSRC Fellowship for Growth (project EP/M002187/1). The authors would like to thank Dr Aivar Sootla and Dr Tianshi Chen for helpful discussions. 
}
\thanks{Supporing Information (SI) can be found online \cite{supp}.}
}

\maketitle

\begin{abstract} 
In this paper, we propose a new method to identify biochemical reaction networks (i.e.  both reactions and kinetic parameters) from heterogeneous datasets. Such datasets can contain (a) data from several replicates of an experiment performed on a biological system; (b) data measured from a biochemical network subjected to different experimental conditions, for example, changes/perturbations in biological inductions, temperature, gene knock-out, gene over-expression, etc. Simultaneous integration of various datasets to perform system identification has the potential to avoid non-identifiability issues typically arising when only single datasets are used.
\end{abstract}

\section{Introduction}
The problem of identifying biological networks from experimental time series data is of fundamental interest in systems and synthetic biology \cite{claire2015cdc}. For example, such information can aid in the design of drugs or of synthetic biology genetic controllers. 
Tools from system identification \cite{ljung1999system} can be applied for such purposes. However, most system identification methods produce estimates of model structures based on data coming from a single experiment.

The interest in identification methods able to handle several datasets simultaneously is twofold.
Firstly, with the increasing availability of ``big data'' obtained from sophisticated biological instruments, e.g. large `omics' datasets, attention has turned to the efficient and effective integration of these data and to the maximum extraction of information from them. Such datasets can contain (a) data from replicates of an experiment performed on a biological system of interest under identical experimental conditions; (b) data measured from a biochemical network subjected to different experimental conditions, for example, different biological inducers, temperature, stress, optical input, gene knock-out and over-expression, etc. The challenges for simultaneously considering heterogeneous datasets during system identification are: (a) the system itself is unknown, i.e. neither the structure nor the corresponding parameters are known; (b) it is unclear how heterogeneous datasets collected under different experimental conditions influence the ``quality'' of the identified system. 

Secondly, in control or synthetic biology applications the systems to be controlled typically need to be modelled first. Highly detailed or complex models are typically difficult to handle using rigorous control design methods. Therefore, one typically prefers to use simple or sparse models that capture at best the dynamics expressed in the collected data. The identification and use of simple or sparse models inevitably introduces model class uncertainties and parameter uncertainties \cite{kaltenbach2009systems, vanlier2013parameter}. To assess these uncertainties replicates of multiple experiments is typically necessary.

Our approach is based on the concept of sparse Bayesian learning \cite{tipping2001sparse, wipf2011latent} and on the definition of a unified optimisation problem allowing the consideration of different parameter values for different experimental conditions, and whose solution is a model consistent with all datasets available for identification. The ability to consider various datasets simultaneously can potentially avoid non-identifiability issues arising when a single dataset is used \cite{ingolia2008systems}. Furthermore, by comparing the identified parameter values associated with different conditions, we can pinpoint the influence specific experimental changes have on system parameters.

The notation in this paper is standard and can be found in SI Section~\ref{app:notation}.

\section{Model}\label{sec:identification}
We consider dynamical systems described by nonlinear differential/difference equation with additive noise: 
\begin{equation}
\begin{aligned}
\delta({x}{_{nt}}) &=\bff_n(\bx_t,\bu_t)\bv_n+\xi{_{nt}}\quad  i =1, \ldots, n_{\bx} \\
&=\sum\nolimits_{s=1}^{N_{n}}v_{ns}f_{ns}(\bx_t,\bu_t)+\xi{_{nt}},  
\label{eq:expansion}
\end{aligned}
\end{equation}
where $\delta({x}{_{nt}}) = \dot{x}{_{nt}}$ for continuous-time system; $\delta({x}{_{nt}}) = {x}{_{nt}}\text{ or } {x}{_{nt}}-{x}{_{n, t-1}}$  or some \emph{known} transformation of historical data for discrete-time system; 
$v_{ns} \in {\mathbb{R}}$ and $f_{ns}(\bx_t,\bu_t): \mathbb{R}^{n_\bx+n_\bu}\rightarrow \mathbb{R}$ are basis functions that govern the dynamics. To ensure existence and uniqueness of solutions, the functions $f_{ns}(\bx_t,\bu_t)$ are assumed to be Lipschitz continuous. 
Note that we do not assume \emph{a priori} knowledge of the form of the nonlinear functions appearing on the right-hand side of the equations in~\eqref{eq:expansion}, e.g. whether the degradation obeys first-order or enzymatic catalysed dynamics or whether the proteins are repressors or activators. 

When data are sampled, we assume the data matrix and first derivative/difference data matrix satisfying (\ref{eq:expansion}) can be obtained as 
\begin{equation}
\begin{aligned}
\begin{bmatrix}	
x_{11} & \ldots & x_{n_\bx 1}  \\ 
\vdots  & \ddots & \vdots  \\
x_{1M} & \ldots & x_{n_\bx M} \\ 
\end{bmatrix}~\text{and}~\begin{bmatrix}	
\delta({x}_{11}) & \ldots & \delta({x}_{n_\bx 1})  \\ 
\vdots  & \ddots & \vdots  \\
\delta({x}_{1M}) & \ldots & \delta({x}_{n_\bx M})  \\ 
\end{bmatrix}
\label{datamatrix}
\end{aligned}
\end{equation}
respectively.

Based on these defined data matrices, the system in (\ref{eq:expansion}) can be written as
$
\mathbf{y}_n=\dicc_n\bv_n+\bXi_n, \ n=1,\ldots,n_{\bx},
$
where
$\by_n \define\left[\delta({x}_{n1}),\ldots,\delta({x}_{n{M}})\right]^\top\in {\mathbb{R}}^{M\times 1}$,
$\bv_n \define \left[v_{n1},\ldots,v_{nN_n}\right]^\top \in {\mathbb{R}}^{N_n\times 1}$,
$\bXi_n \define
\left[\xi_{n1},\ldots,\xi_{nM})\right]^\top\in {\mathbb{R}}^{M\times 1}$, and the dictionary matrix $\mathbf{\dicc}_n\in {\mathbb{R}}^{M\times N_n} $ with its $j$-th column being $[f_{nj}(\bx_1,\bu_1), \ldots, f_{nj}(\bx_M,\bu_M)]^{\top}$.
The noise or disturbance vector $\bXi_n$ is assumed to be Gaussian distributed with zero mean and covariance $\Cov \in \R_{+}^{M \times M}$ \footnote{Note that the covariance matrix is not necessarily diagonal.}.
The identification goal is to estimate $\bv_n$ of the linear regression formulation $\mathbf{y}_n=\dicc_n\bv_n+\bXi_n, \ n=1,\ldots,n_{\bx}$. Two issues need to be raised here. 
The first one is the selection of basis function $f_{ns}(\cdot,\cdot)$ which is key to the success of identification. Some discussion on this can be found in SI Section~\ref{app:basis}, especially for biochemical reaction networks.
The second one is the estimation of the first derivative data matrix which is not trivial. In SI Section~\ref{app:deri} , we provide a method to estimate first derivatives from noisy time-series data.

If a total number of $C$ datasets are collected from $C$ independent experiments, we put a subscript $\left[c\right]$ to index the identification problem associated with the specific dataset obtained from experiment $[c]$. 
In what follows we gather in a matrix $\Anc$ similar to $\dicc_n$ the set of \emph{all} candidate/possible  basis functions that we want to consider during the identification. The identification problem is then written as:
\begin{equation}
\begin{aligned}
\ync=\Anc\wnc+\noisenc,  \ \ n=1,\ldots,n_{\bx}, \ c = 1, \ldots, C. 
\label{problem:0}
\end{aligned}
\end{equation}
The solution to $\wnc$ to \eqref{problem:0} is typically going to be sparse, which is mainly due to the potential introduction of non-relevant and/or non-independent basis functions in $\Anc$.

Since the $n_{\bx}$ linear regression problems in (\ref{problem:0}) are independent, for simplicity of notation, we omit the subscript $n$ used to index the state variable and simply write:
\begin{equation}
\yc=\Ac\wc+\noisec, c = 1, \ldots, C,
\label{problem:single}
\end{equation}
in which, 
\begin{equation}
\begin{aligned}
\Ac &\define \left[\Ac_{:,1}, \ldots, \Ac_{:,N}\right] \\
& =
\left[
\begin{array}{ccc}
f_{1}(\bx^{[c]}_1) &\ldots & f_{N}(\bx^{[c]}_1) \\
\vdots  & & \vdots \\
f_{1}(\bx^{[c]}_{{M^{[c]}}})&\ldots & f_{N}(\bx^{[c]}_{{M^{[c]}}})
\end{array}
\right]  \in \R^{{M^{[c]}} \times N}, \\
\wc &\define \left[w_1^{\left[c\right]}, \ldots, w_N^{\left[c\right]}\right]^\top \in \R^N,\\
\boldsymbol{\xi}^{\left[c\right]} &\triangleq  \left[\xi^{\left[c\right]}_1, \ldots,  \xi^{\left[c\right]}_{{M^{[c]}}})\right]^\top \in \R^{M^{[c]}},
\end{aligned}
\end{equation}
where $\bx^{[c]}_t= \left[x^{[c]}_{1t}, \ldots, x^{[c]}_{n_{\bx} t}\right] \in \R^{n_{\bx}}$ is the state vector at time instant $t$.
It should be noted that $N$, the number of basis functions or number of columns of the dictionary matrix $\Ac \in \mathbb{R}^{{M^{[c]}} \times N}$, can be very large.	
Without loss of generality, we assume $M^{[1]} = \cdots=M^{[C]} = M$.

{\it
\begin{remark}
The model class considered in~\eqref{eq:expansion} can be enlarged in various ways. First, measurement noise, which is ubiquitous in practice, can be accounted for using the following linear measurement equation:
\begin{equation}
\begin{aligned}
z_t = x_t+\epsilon_t,
\label{eq:measurement}
\end{aligned}
\end{equation} where the measurement noise $\epsilon_t$ is assumed i.i.d. Gaussian. Under this formulation, the noise-contaminated data $z_t$ represents the collected data rather than $x_t$ in~\eqref{datamatrix}.
Second, the additive stochastic term $\xi_t$ in \eqref{eq:expansion} is often used to model dynamic noise or diffusion. In many practical application, however, it is necessary to account for multiplicative noise instead of additive noise. Multiplicative noise can be accounted for by replacing \eqref{eq:expansion} with $\dot{x}_{t} =f(\bx_t,\bu_t)\bv+h(\bx_t,\bu_t)\xi{_{t}}.$
In SI Section~\ref{app:noise}, we show how the framework presented here can be modified to encompass these extensions. 
\end{remark}
}

\section{Identification from multiple datasets}
\label{sec:multiple}

To ensure reproducibility, experimentalists repeat their experiments under the same conditions, and the collected data are then called ``replicates''. Typically, only the average value over these replicates is used for modelling or identification purposes. In this case, however, only the first moment is used and information provided by higher order moments is lost. Moreover, when data originate from different experimental conditions, it is usually very hard to combine the datasets into a single identification problem.  This section will address these issues by showing how several datasets can be combined to define a unified optimisation problem whose solution is an identified model consistent with the various datasets available for identification.

To consider heterogeneous datasets in one single formulation, we stack the various equations in \eqref{problem:single} (see Eq.~\eqref{problem:stack}). Each stacked equation in Eq.~\eqref{problem:stack} corresponds to a replicate or an experiment performed by changing the experimental conditions on the same system.
\begin{figure*}
{\small
\begin{equation}
\begin{aligned}
\left[
\begin{array}{c}
\by^{\left[1\right]}\\
\vdots \\
\by^{\left[C\right]}
\end{array}
\right] 
& =
\myunderbrace{
\left[
\begin{array}{ccc|c|ccc}
\dic ^{\left[1\right]} _{:,1} &\ldots & \dic ^{\left[1\right]}_{:,N} &  & & &   \\
&  & &    \ddots  & &  &\\
& &  & &    \dic ^{\left[C\right]}_{:,1} &\ldots & \dic ^{\left[C\right]}_{:,N}
\end{array}
\right]
}{C~\textbf{Blocks}}
\left[
\begin{array}{c}
\bw^{\left[1\right]}\\\hline
\vdots \\ \hline
\bw^{\left[C\right]}
\end{array}
\right]
+
\left[
\begin{array}{c}
\bxi^{\left[1\right]}\\\hline
\vdots \\\hline
\bxi^{\left[C\right]}
\end{array}
\right]\\
&=
\myunderbrace{
\left[
\begin{array}{ccc|c|ccc}
\dic^{\left[1\right]}_{:,1}& &  & &  \dic^{\left[1\right]}_{:,N}& &\\
& \ddots & &    \ddots  & & \ddots &\\
& & \dic ^{\left[C\right]}_{:,1}& &   & &  \dic^{\left[C\right]}_{:,N}
\end{array}
\right]
}{N ~\textbf{Blocks}}
\left[
\begin{array}{c}
w_1^{\left[1\right]} \\ \vdots \\ w_1^{\left[C\right]} \\ \hline
\vdots\\ \hline
w_N^{\left[1\right]}\\ \vdots \\ w_N^{\left[C\right]}
\end{array}
\right]
+
\left[
\begin{array}{c}
\bxi^{\left[1\right]}\\\hline
\vdots \\\hline
\bxi^{\left[C\right]}
\end{array}
\right] = \left[
\begin{array}{c|c|c}
\dic_1&\cdots &\dic_N
\end{array}
\right]
\left[
\begin{array}{c}
\bw_1\\\hline
\vdots \\\hline
\bw_N
\end{array}
\right]
+
\left[
\begin{array}{c}
\bxi^{\left[1\right]}\\\hline
\vdots \\\hline
\bxi^{\left[C\right]}
\end{array}
\right].
\label{problem:stack}
\end{aligned}
\end{equation}}
\end{figure*}

In Eq.~\eqref{problem:stack}, $\dic_i = \blkdiag[\dic^{\left[1\right]}_{:,i}, \ldots, \dic^{\left[C\right]}_{:,i}]$, and $\bw_i = [w_i^{\left[1\right]}, \ldots, w_i^{\left[C\right]}]^\top$, for $i = 1, \ldots, N$. Based on the stacked formulation given in Eq.~\eqref{problem:stack} we further define
\begin{equation}
\begin{aligned}
\by &= \left[
\begin{array}{c}
\by^{\left[1\right]}\\
\vdots \\
\by^{\left[C\right]}\\
\end{array}
\right], 
\dic =
\left[
\begin{array}{c|c|c}
\dic_1&\cdots &\dic_N
\end{array}
\right],\\
\bw &= \left[
\begin{array}{c}
\bw_1\\\hline
\vdots \\\hline
\bw_N
\end{array}
\right],
\bxi =
\left[
\begin{array}{c}
\bxi^{\left[1\right]}\\\hline
\vdots \\\hline
\bxi^{\left[C\right]}
\end{array}
\right],
\label{problem:stack:definition}
\end{aligned}
\end{equation}
which gives  
\begin{equation}
\begin{aligned}
\by = \dic \bw +\bxi. 
\label{problem}
\end{aligned}
\end{equation}
This yields a formulation very similar to that presented previously for a single linear regression problem. However, in the multi-experiment formulation \eqref{problem}, there is now a special block structure for $\by$, $\dic$ and $\bw$. 

{\it
\begin{remark}
	When $\bw^{\left[i\right]}$ is fixed to be $\bw $ for all the experiments, i.e. $ \bw^{\left[1\right]}= \cdots = \bw^{\left[C\right]} = \bw$, we can formulate the identification problem as a single linear regression problem by concatenation:
{\small
	\begin{equation}
	\begin{aligned}
	\left[
	\begin{array}{c}
	\by^{\left[1\right]}\\
	\vdots \\
	\by^{\left[C\right]}
	\end{array}
	\right] &=
	\left[
	\begin{array}{c}
	\dic^{\left[1\right]}\\
	\vdots \\
	\dic^{\left[C\right]}
	\end{array}
	\right]
	\bw
	+
	\left[
	\begin{array}{c}
	\bxi^{\left[1\right]}\\
	\vdots \\
	\bxi^{\left[C\right]}
	\end{array}
	\right].
	\label{problem:cat}
	\end{aligned}
	\end{equation}
}
\end{remark}
}
 
To incorporate prior knowledge into the identification problem, it is often important to be able to impose constraints on $\bw$. In biological systems, positivity of the parameters constituting $\bw$ is an example of such constraints.
The other motivation for constrained optimisation comes from stability considerations. Typically, the underlying system is known \emph{a priori} to be stable, especially if this system is a biological or physical system. A lot of stability conditions can be formulated as convex optimisation problems, e.g. Lyapunov stability conditions expressed as Linear Matrix Inequalities \cite{boyd1987linear}, Gershgorin Theorem for linear systems \cite{horn1990matrix}, etc. Only few contributions are available in the literature that address the problem of how to consider \emph{a priori} information on system stability during system identification \cite{cerone2011enforcing,zavlanos2011inferring}.
To be able to integrate constraints on $\bw$ into the problem formulation, we consider the following assumption on $\bw$.
\begin{assumption}
\label{assumption-constraints}
Constraints on the weights $\bw$ can be described by a set of convex functions:
$H^{[I]}_{i}(\bw)\leq0$, $\forall i$; $H^{[E]}_{j}(\bw)=0$, $\forall j$,
where the convex functions $H^{[I]}_{i}: \R^{N}\rightarrow \R$ are used to define inequality constraints, whereas the convex functions $H^{[E]}_{j}: \R^{N}\rightarrow \R$ are used to define equality constraints.
\end{assumption}

\section{Methods}
To get an estimate of $\bw$ in~\eqref{problem}, we use Bayesian modelling to treat all unknowns as stochastic variables with certain probability distributions \cite{bishop2006pattern}. For $\by=\dic \bw+\bXi$, it is assumed that the stochastic variables in the vector $\bXi$ are Gaussian distributed with \emph{unknown} covariance matrix $\Cov$, i.e., $\bXi\thicksim\bN(\mathbf{0}, \Cov)$. 

In what follows we consider the following variable substitution for the inverse of unknown covariance matrix or precision matrix:
$\Covinv \define \Cov^{-1}.$ 
In such case, the likelihood of the data given $\bw$ is
{\small\begin{equation}
\begin{aligned}
\Prob(\by|\bw)
&=\mathcal{N}(\by|{\dic} {\bw},\Cov) \propto\exp \left[ -\frac{1}{2}(\dic\bw-\by)^{\top}\Covinv (\dic\bw-\by)\right].
\label{Likelihood}
\end{aligned}
\end{equation}}
\subsection{Sparsity Inducing Priors}
In Bayesian models, a prior distribution $\Prob(\bw)$ can be defined as
$\Prob(\bw) = \prod_{i = 1}^{N} \Prob(\bw_i)$
where
$
\Prob(\bw_i)\propto\exp \left[-\frac{1}{2}\sum_{j=1}^{C}g(w_i^{\left[j\right]})\right]=\prod_{j=1}^{C}\exp \left[-\frac{1}{2}g(w_i^{\left[j\right]})\right]=\prod_{j=1}^{C}\Prob(w_i^{\left[j\right]}),
$
with $g(w_i^{\left[j\right]})$ being a given function of $w_i^{\left[j\right]}$. Generally, $\bw$ in~\eqref{problem} is sparse, and therefore certain sparsity properties should be enforced on $\bw$. To this effect, the function $g(\cdot)$ is usually chosen to be a concave, non-decreasing function of $|w_i^{\left[j\right]}|$ \cite{wipf2011latent}. Examples of such functions $g(\cdot)$ include Generalised Gaussian priors and Student's \emph{t} priors (see \cite{palmer2006variational,wipf2011latent} for details).

Computing the posterior mean $\E(\bw|\by)$ is typically intractable because the posterior $\Prob(\bw|\by)$ is highly coupled and non-Gaussian. To alleviate this problem, ideally one would like to approximate $\Prob(\bw|\by)$ as a Gaussian distribution for which efficient algorithms to compute the posterior exist \cite{bishop2006pattern}. For this, the introduction of lower bounding \emph{super-Gaussian} priors $\Prob(w_i^{\left[j\right]})$, i.e., 
$
\Prob(w_i^{\left[j\right]})  =\max_{\gamma_{i} >0}\bN(w_i^{\left[j\right]}|0,\gamma_{i})\hyperprior(\hyper_i)
\label{singlepriors},
$
can be used to obtain an analytical approximation of $\Prob(\bw|\by)$ \cite{palmer2006variational}.

Note that problem~\eqref{problem} has a block-wise structure, i.e. the solution $\bw$ is expected to be block-wise sparse.
Therefore, sparsity promoting priors should be specified for $\Prob(\bw_i)$, $\forall i$.
To do this, for each block $\bw_i$, we define a hyper-parameter $\hyper_i$ such that
\begin{equation}
\begin{aligned}
\Prob(\bw_i) & =\max_{\gamma_{i} >0}\bN(\bw_i|\mathbf{0},\gamma_{i}\bI_C)\hyperprior(\hyper_i) \\
&=\max_{\gamma_{i} >0}\prod_{j=1}^{C}\bN(w_i^{\left[j\right]}|0,\gamma_{i})\hyperprior(\hyper_i),
\label{priors}
\end{aligned}
\end{equation}
where $\prior(\hyper_i)$ is a nonnegative function, which is treated as a hyperprior with $\hyper_i$ being its associated hyperparameter. Throughout, we call $\prior(\hyper_i)$ the ``\emph{potential function}''. This Gaussian relaxation is possible if and only if $\log \Prob(\sqrt{w_i})$ is concave on $(0,\infty)$.
Defining
\begin{equation}
\begin{aligned}
\Hyper_{i} &= \left[\gamma_{i}, \ldots, \gamma_{i} \right]\in \R^{C}, \ \bG_{i} =\diag\left[\Hyper_{i}\right], \\
\Hyper&= \left[\Hyper_{1}, \ldots, \Hyper_{\n}\right]\in \R^{\n C}, \ \bG = \diag\left[\Hyper\right],
\label{eq:definegamma}
\end{aligned}
\end{equation}
we have
{\small
\begin{equation}
\begin{aligned}
\Prob(\bw)  = \prod_{i=1}^{\n} \Prob(\bw_i) =\max_{\Hyper > \mathbf{0}}
\bN(\bw|\mathbf{0},\bG)
\hyperprior(\Hyper).
\label{Prior}
\end{aligned}
\end{equation}
}

\subsection{Cost Function}
\label{sec:costfunction}
Once we introduce the Gaussian likelihood in \eqref{Likelihood} and the variational prior in \eqref{Prior}, we can get the following optimisation problem jointly on $\bw$, $\bgamma$ and $\Covinv$.
\begin{proposition}
The unknowns $\bw, \bgamma, \Covinv$ can be obtained by solving the following optimisation problem
{\small
\begin{equation}
\begin{aligned}
& \mathcal{L}(\bw, \bgamma, \Covinv) 
= \min_{ \bw, \bgamma, \Covinv} \{- \log |\Covinv| +\log  |\bG|
+\log |\bG^{-1}+\dic^\top\Covinv \dic |\\
&+ \MSE+ \bw^\top\bG^{-1}\bw +\sum_{j=1}^{N}p(\hyper_{j}) \},
\label{eq:cost}
\end{aligned}
\end{equation}
}
where $\bG$ is given in \eqref{eq:definegamma}. 
\end{proposition}

\begin{proof}
The derivation can be found in the SI \ref{app:cost}. The proof mainly relies on using marginal likelihood maximisation.
\end{proof}

\subsection{Algorithm}

The cost function in~\eqref{eq:cost} is convex in $\bw$ and $\Covinv$ but concave in $\bGamma$.
This non-convex optimisation problem can be formulated as a convex-concave procedure (CCCP). It can be shown that solving this CCCP is equivalent to solving a series of iterative convex optimisation programs, which converges to a stationary point~\cite{sriperumbudur2009convergence}.
Let
\begin{alignat}{2}
 u(\bw, \bgamma, \Covinv) &\define  \MSE
+\bw^\top\bG^{-1}\bw-\log \det\Covinv,\notag \\	
v(\bgamma,\Covinv) &\define -\left[ \log  |\bG|+\log | \bG^{-1}+ \dic^\top\Covinv \dic | +\sum_{j=1}^{N}p(\hyper_j)\right]. \notag
\end{alignat}

It is easy to check that $v(\bgamma,\Covinv)$ is a convex function with respect to~$\bgamma$. Furthermore, $\log|\cdot|$ is concave in the space of positive semi-definite matrices.
Since we adopt a super-Gaussian prior with potential function $\prior(\hyper_j), \forall j,$ as described in \eqref{priors}, a direct consequence is that $p(\hyper_j)=-\log\prior(\hyper_j)$ is concave, and, therefore, $-p(\hyper_j)$ is convex \cite{tipping2001sparse}.\footnote{In this paper, the prior is chosen as a Student's t prior thus $p(\hyper_j) = 1$.} Note that $u(\bw,\bgamma,\Covinv)$ is jointly convex in $\bw$, $\bgamma$ and $\Covinv$, while $v(\bgamma, \Covinv)$ is jointly convex in $\bgamma$ and $\Covinv$.
As a consequence, the minimisation of the objective function can be formulated as a  concave-convex procedure
\begin{equation}
\begin{aligned}
\label{cccp}
\min_{\bgamma\succeq\mathbf{0},\Covinv\succeq \mathbf{0}, \bw} u(\bw, \bgamma,\Covinv)-v(\bgamma, \Covinv).
\end{aligned}
\end{equation}
Since $v(\bgamma, \Covinv) $ is differentiable over $\bgamma$, the problem in (\ref{cccp}) can be transformed into the following iterative convex optimisation problem
{\small
\begin{align}
	\bw^{k+1}
	&=\argmin\limits_{\bw} u(\bw, \bgamma^k, \Covinv^k) \label{eq:cccp1}\\
	\bgamma^{k+1}
	&=\argmin\limits_{\bgamma \succeq \mathbf{0}} u(\bw^k, \bgamma,\Covinv^k)-\nabla_{\bgamma} v(\bgamma^k, \Covinv^k)^\top\bgamma \label{eq:cccp2}\\
	\Covinv^{k+1}
	&=\argmin\limits_{\Covinv\succeq \mathbf{0}}u(\bw^k, \bgamma^k,\Covinv)-\nabla_{\Covinv} v(\bgamma^k,\Covinv^k)^\top\Covinv.
	\label{eq:cccp3}
\end{align}
}

Using basic principles in convex analysis, we then obtain the following analytic form for the negative gradient of $v(\bgamma)$ at $\bgamma$ (using the chain rule):
 
{\small 
\begin{equation}
\begin{aligned}
\balpha^{k} 
\triangleq&-\nabla_{\bgamma} v(\bgamma,\Covinv^k)^\top  |_{\bgamma=\bgamma^k}\\
=&\nabla_{\bgamma}\left[\log |\bG^{-1}+\dic^\top \Covinv^k \dic|
+\log |\bG| \right]\\
=& \diag \{\left[-(\bG^{k})^{-1}+ \dic^\top \Covinv^k \dic\right]^{-1}\}\cdot \diag\{-(\bG^{k})^{-2}\} \\
& +\diag^{-1}\{\bG^k\} \\
=& \myunderbrace{
\left[
\begin{array}{c|c|c}
\balpha_{11}^{k}& \cdots &  \balpha_{1\n}^{k}
\end{array}
\right]
}{N ~\textbf{Blocks}} \\ \\
=& \left[
\begin{array}{c|c|c}
\myunderbrace{\alpha_{11}^{k}, \ldots, \alpha_{11}^{k}
}{C ~\textbf{Elements}}
& \cdots &  	
\myunderbrace{\alpha_{1\n}^{k}, \ldots, \alpha_{1\n}^{k}
}{C ~\textbf{Elements}}
\end{array}
\right].\\
\label{alpha1k}
\end{aligned}
\end{equation}
}

Therefore, the iterative procedures~\eqref{eq:cccp1} and~\eqref{eq:cccp2} for $\bw^{k+1}$ and $\bgamma^{k+1}$, respectively, can be formulated as
\begin{equation}
\begin{aligned}
\left[\bw^{k+1},\bgamma^{k+1}\right] =\argmin\limits_{\bgamma\succeq \mathbf{0},\bw}
& \MSEk\\
&+\sum_{i=1}^{\n} \left(\frac{\bw_i^\top\bw_i}{\gamma_{i}} +C\gamma_{i}\alpha_{i}^{k}\right).
\label{cccp-4}
\end{aligned}
\end{equation}
The optimal $\bgamma$ components are obtained as:
$
\gamma_{i}=\frac{\|\bw_i\|_2}{\sqrt{C \alpha_{i}^{k}}}.
$
If $\bgamma$ is fixed, we have $\bw^{k+1}$ by solving optimisation problem
\begin{equation}
\begin{aligned}
\min\limits_{\bw} &
\MSEk +2\sum_{i=1}^{\n} \|\theta_{i}^k \cdot \bw_i\|_2, 
\label{rwglasso}
\end{aligned}
\end{equation}
where
$\theta_{i}^k   = C\alpha_{i}^{k}$.
We can then inject this into the expression of $\gamma_i$,
which yields
\begin{equation}
\begin{aligned}
\gamma_{i}^{k+1}&=\frac{\|\bw_i^{k+1}\|_2}{\sqrt{C\alpha_{i}^{k}}}. 
\label{gammaik+1}
\end{aligned}
\end{equation}

After we get $\bw^{k+1}$ and $\bgamma^{k+1}$, we can proceed with the optimisation iteration in~\eqref{eq:cccp3}:
\begin{equation}
\begin{aligned}
\Lambda^k 
&=-\nabla_{\Covinv} v(\bgamma^k,\Covinv^k) \\
& =\nabla_{\Covinv}\left(\log \det \left(\bG^{-k}+ \dic^\top\Covinv^{k} \dic \right) \right)\\
& = \dic(\bG^{-k}+\dic^\top\Covinv^{k}  \dic )^{-1}\dic^\top.
\label{WeightforCov}
\end{aligned}
\end{equation}
Letting $\mathbf{Y}^{k+1} = (\dic\bw^{k+1}-\by)\cdot(\dic\bw^{k+1}-\by)^\top$, we can get an estimate of the inverse of covariance matrix $\Covinv$ as:
\begin{equation}
\begin{aligned}
\Covinv^{k+1} =\argmin\limits_{\Covinv \succeq \mathbf{0}}
\trace\left(\Covinv\mathbf{Y}^{k+1} \right)-\log\det \Covinv
+\trace\left(\Lambda^k  \Covinv\right).
\label{covink+1}
\end{aligned}
\end{equation}
Given $\bgamma^{k+1}$ in~\eqref{gammaik+1} and $\Covinv^{k+1}$ in~\eqref{covink+1}, we can then go back to~\eqref{alpha1k} to update $\balpha$ for the next iteration.

This above described iterative procedure for identification is summarised in Algorithm~\ref{alg:summary} below. 
{
\begin{algorithm}[!]
\caption{Nonlinear Identification Algorithm using Heterogeneous Datasets}
\label{alg:summary}
\begin{algorithmic}[1]
\STATE Collect $C$ heterogeneous groups of time series data from the system of interest (assuming the system can be described by~\eqref{eq:expansion});
\STATE Select the candidate basis functions that will be used to construct the dictionary matrix described in Section~\ref{sec:multiple};
\STATE Initialise $\theta_i^0=1, \ \forall i$, $\alpha_i^0 =\frac{\theta_i^0}{C}$, $\Covinv^{0}= \bI$, $\Lambda^{0}= \bI$;
\FOR  {$k=0, \ldots, k_{\max}$}
\STATE $\bw^{k+1}$ can be obtained by solving the following weighted minimisation problem over $\bw$, subject to the convex constraints in Assumption~\ref{assumption-constraints}
{\small 
\begin{equation}
\begin{aligned}
\min\limits_{\bw} 
\frac{1}{2}\MSEk +\sum_{i=1}^{\n} \|\theta_{i}^k \cdot \bw_i\|_2 ;
\label{alg:rwglasso}
\end{aligned}
\end{equation}
}
\STATE Update $\gamma_{i}^{k+1}$ using~\eqref{gammaik+1};
\STATE Let $\mathbf{Y}^{k+1} = (\dic\bw^{k+1}-\by)\cdot(\dic\bw^{k+1}-\by)^\top$;
\STATE $\Covinv^{k+1} $ can be obtained by solving the following weighted minimisation problem over the inverse of the covariace matrix:
\begin{equation}
\begin{aligned}
\min\limits_{\Covinv \succeq \mathbf{0}}
\trace\left(\mathbf{Y}^{k+1} +\Lambda^k \right)\Covinv-\log\det \Covinv;
\label{alg:rwcovariance}
\end{aligned}
\end{equation}

\STATE Update $\balpha^{k+1}$ using~\eqref{alpha1k};

\STATE Update $\theta_{i}^{k+1} =  C\alpha_{i}^{k+1}$; 

\STATE Update $\Lambda^{k+1}$ using~\eqref{WeightforCov};

\IF   {a stopping criterion is satisfied}
\STATE Break;
\ENDIF
\ENDFOR
\end{algorithmic}
\end{algorithm}
}
Some further discussion can be found in SI Section~\ref{s:discussion}.

\subsection{ADMM Implementation}
\label{sec:admm_sub}
Essentially, Algorithm~\eqref{alg:summary} consists of a reweighted Group Lasso algorithm~\eqref{alg:rwglasso} and a reweighted inverse covariance estimation algorithm~\eqref{alg:rwcovariance}.
Algorithm~\eqref{alg:summary} can be implemented using the Alternating Direction Method of Multipliers (ADMM) \cite{boyd2011distributed}. This ADMM parallelisation allows to distribute the algorithmic complexity to different threads and to build a platform for scalable distributed optimisation. This is key to be able to deal with problems of large dimensions. More details can be found in SI Section~\ref{s:admm}.
\subsection{Connection to SDP formulations and the sparse multiple kernel method}
\label{sec:sdp}
The iteration in \eqref{cccp-4} can be rewritten in the following compact form
\begin{equation}
\begin{aligned}
\left[\bw^{k+1},\bgamma^{k+1}\right] =\argmin\limits_{\bgamma \succeq \mathbf{0}, \bw}
& \MSEk \\
&+ \bw^{\top} \bG^{-1}  \bw -\nabla_{\bgamma} v(\bgamma^k, \Covinv^k)^\top \bgamma.
\label{cccp-compact}
\end{aligned}  
\end{equation}
This is equivalent to the following SDP by using the standard procedure in \cite{boyd2004convex}.
\begin{equation}
\begin{split}
\min_{\bz, \bw, \Hyper}\,\,\,\,\,\, & \, \bz -\nabla_{\bgamma} v(\bgamma^k, \Covinv^k)^\top \bgamma \notag\\
\mathrm{subject}\,\,\mathrm{to}\,\,\,\,\,\,\, &
\begin{bmatrix}
\bz & (\by-\dic \bw)^{\top} & \bw ^{\top} \\
\by-\dic \bw & (\Covinv^{k})^{-1} & \mathbf{0} \\
 \bw  & \mathbf{0}  & \bG
\end{bmatrix} \succeq \mathbf{0}\\
& \,\,\,  \bgamma \succeq \mathbf{0}
\end{split}
\end{equation}
The cost of solving this SDP is at least $N^3$ as well as $M$. Therefore, solving this SDP is too costly for all but problems with a small number of variables. This means that the number of samples, the dimension of the system, etc., can not be too large simultaneously. In this SDP formulation, $\bG$ is closely related to the sparse multiple kernel presented in \cite{TianshiTAC}. Certain choice of kernels may introduce some good properties or help reduce algorithmic complexity. In our case, we choose $\bG$ to have a diagonal or a DC kernel structure.

\section{Simulations}
In this section, we use numerical simulations to show the effectiveness of the proposed algorithm. 
To compare the identification accuracy of the various algorithms considered, we use the root of normalised mean square error (RNMSE) as a performance index, i.e. $\textbf{RNMSE} = \|\bw_{\text{estimate}}-\bw_{\text{true}}\|_2/\|\bw_{\text{true}}\|_2.$
Several factors affect the RNMSE, e.g. number of experiments $C$, measurement noise intensity, dynamic noise intensity, length of single time series data $M$, number of candidate basis functions $N$. 
For brevity of exposition, we only show results pertaining to change of RNMSE over number of experiment $C$ and length of single time series for one experiment, all in the noiseless case. More results related to other factors that may affect RNMSE will be shown in a future journal publication presenting these results in more details.

As an illustrative example, we consider a model of an eight species generalised repressilator \cite{strelkowa2010switchable}, which is a system where each of the species represses another species in a ring topology. The corresponding dynamic equations are as follows: 
\begin{equation}
\label{example:equations}
\begin{aligned}
\dot x_{1t} &= \frac{p_{11}}{p_{12}^{p_{13}} + x_{8t}^{p_{13}}} + p_{14} - p_{15} x_{1t}, \\
\dot x_{it}&= \frac{p_{i1}}{p_{i2}^{p_{i3}} + x_{i-1,t}^{p_{i3}}} + p_{i4} - p_{i5} x_{it},~\forall i = 2,\dots 8,  
\end{aligned}
\end{equation}
where $p_{ij}$, $i = 1, \ldots, 8, ~j = 1, \ldots, 5$. 
We assume the mean value for these parameters across different species and experiments are $\bar{p}_{i1} = 40$, $\bar{p}_{i2} = 1$, $\bar{p}_{i3} = 3$, $\bar{p}_{i4} = 0.5$, $\bar{p}_{i5} = 1$, $\forall i$. We simulate the ODEs in \eqref{example:equations} to generate the time series data. 
In each ``experiment'' or simulation of \eqref{example:equations}, the initial conditions are randomly drawn from a standard uniform distribution on the open interval $(0,1)$. The parameters in each experiment vary no more than $20\%$ of the mean values. In MATLAB, one can use \texttt{$\bar{p}_{ij}$*(0.8 + 0.4*rand(1))} to generate the corresponding parameters for each experiment.

The numerical simulation procedure can be summarised as follows:
\begin{enumerate}
\item The deterministic system of ODEs~\eqref{example:equations} is solved numerically with an adaptive fourth-order Runge-Kutta method; 
\item As explained in~\eqref{eq:measurement}, Gaussian measurement noise with variance $\sigma^2$ is added to the corresponding time-series data obtained in the previous step\footnote{In the example presented here, we consider the noiseless case corresponding to $\sigma = 0$.};
\item The data is re-sampled with uniform intervals\footnote{In this example, interval length is set to $1$.};
\item The local polynomial regression framework in \cite{de2013derivative} is applied to estimate the first derivative;
\item A dictionary matrix is constructed as explained in Section~\ref{sec:multiple};
\item Algorithm~\ref{alg:summary} is used to identify the model.
\end{enumerate}

Following the procedure described in Section II, the candidate dictionary matrix $\dic$ in step 5) above is constructed by selecting as candidate nonlinear basis functions typically used to represent terms appearing in ODE models of Gene Regulatory Networks.
As a proof of concept, we only consider Hill functions as potential nonlinear candidate functions. The set of Hill functions with Hill coefficient $h$, both in activating and repressing form, for the $i^{th}$ state variables at time instant $t_k$ are:
\begin{equation}
\begin{aligned}
\text{hill}(x_{it}, K, h_{\text{num}}, h_{\text{den}}) &\define \frac{x_{it}^{h_{\text{num}}}}{K^{h_{\text{den}}}+x_{it}^{h_{\text{den}}}}
\end{aligned}
\end{equation}
where $h_{num}$ and $h_{den}$ represent the Hill coefficients. When $h_{\text{num}} = 0$, the Hill function has a repression form, whereas an activation form is obtained for $h_{\text{num}}=h_{\text{den}}\neq 0$.

In our identification experiment, we assume $h_{num}$, $h_{den}$ and $K$ to be known. We are interested in identifying the regulation type (linear or Hill type, repression or activation) and the corresponding parameters $p_{i1}$, the degradation rate constant $p_{i4}$ and the basal expression rate $p_{i5}$, $\forall i$.   Since there are $8$ state variables, we can construct the dictionary matrix $\dic$ with $8$ (basis functions for linear terms) $+ (2*8)$ (basis functions for Hill functions, both repression and activation form) $+1$ (constant unit vector) $=25$ columns. The corresponding matrix $\dic$ is given in Eq.~\eqref{oscillator-dic} in Supporting Information Section~\ref{app:simulation}.

For a fixed number of experiments $C$ and length of single time series $M$, we compute the RNMSE over $50$ simulations by varying initial conditions and parameters $p_{ij}$. The RNMSE over $C$ and $M$ are shown in Fig.~\ref{fig:rnmse_w_iter_1} and  Fig.~\ref{fig:rnmse_w_iter_end}, using both group Lasso and Algorithm~\ref{alg:summary} with the maximal iteration number $k_{\max} = 5$ (see line 4 in Algorithm~\ref{alg:summary}). Inspection of the results presented in Fig.~\ref{fig:rnmse_w_iter_1} and  Fig.~\ref{fig:rnmse_w_iter_end} clearly show that Algorithm~\ref{alg:summary} outperforms significantly group Lasso in terms of RNMSE.

\begin{figure}
\centering
\subfigure[Group Lasso  (first iteration of Algorithm~\ref{alg:summary}). The minimal RNMSE is around $0.75$]{\label{fig:rnmse_w_iter_1}
\includegraphics[scale =0.2]{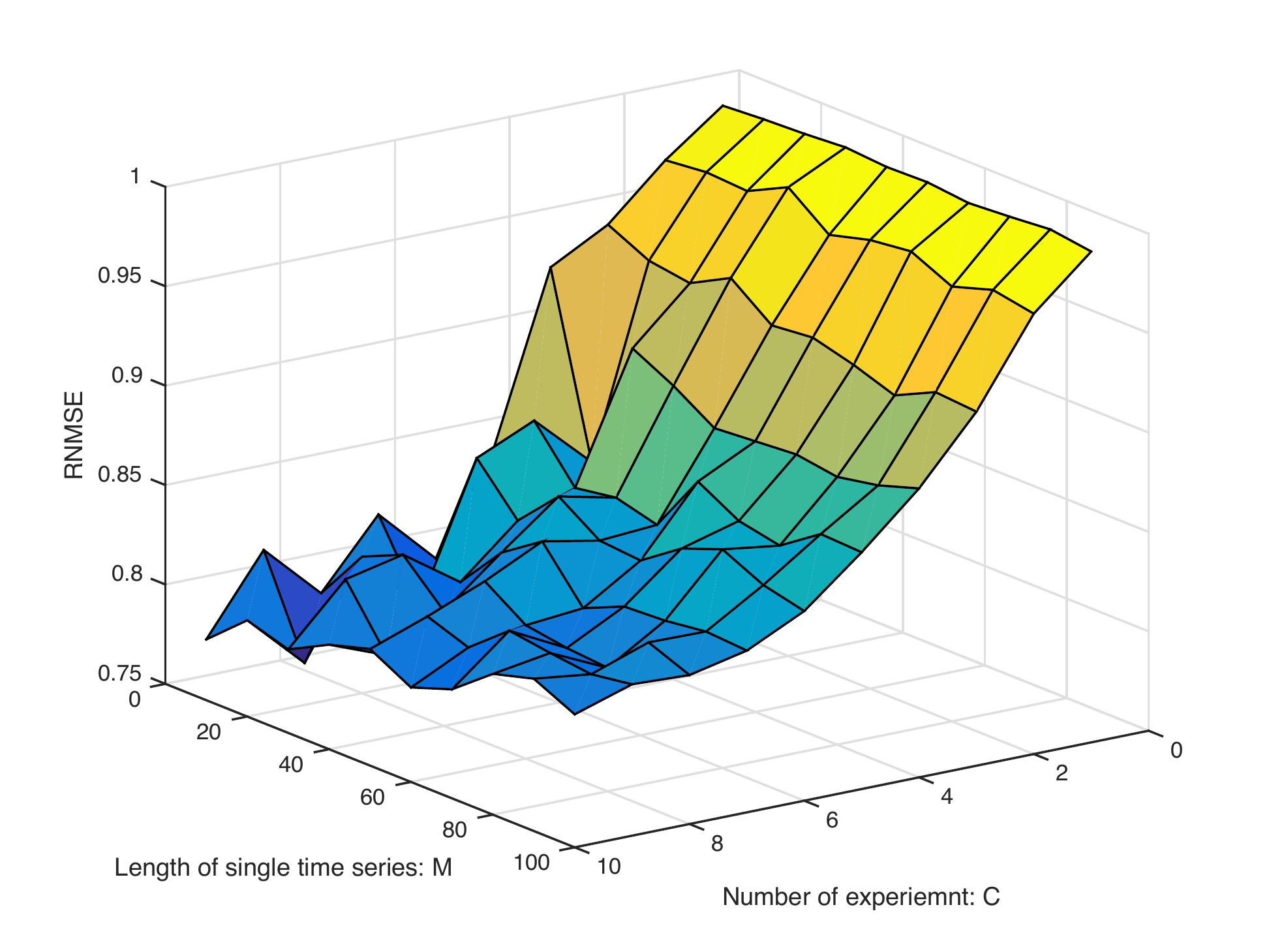}}
\subfigure[ Algorithm~\ref{alg:summary} with maximal iteration number $k_{\max} = 5$. The minimal RNMSE is almost $0$.
]{\label{fig:rnmse_w_iter_end}
\includegraphics[scale =0.2]{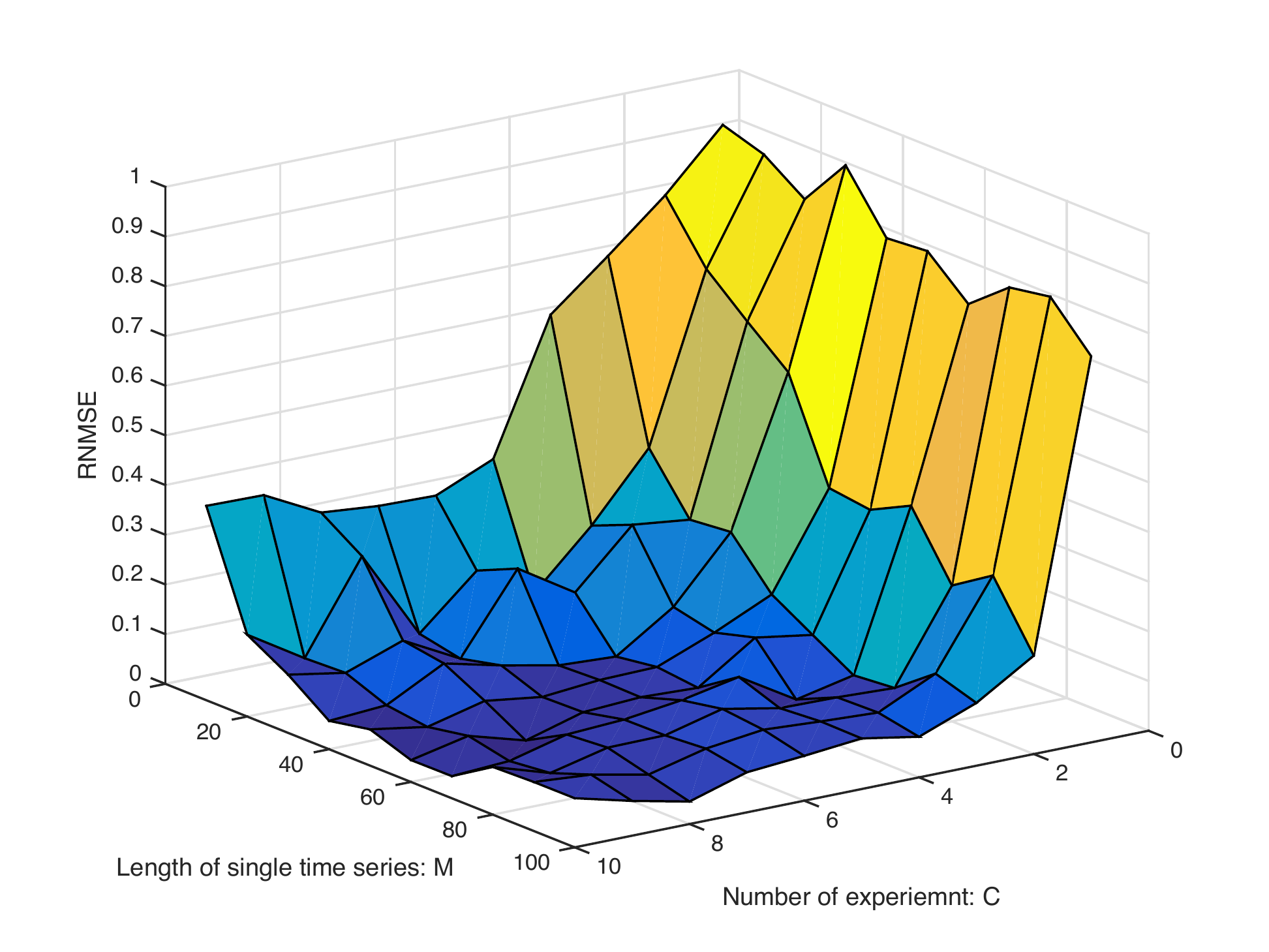}}
\caption[A set of four subfigures.]{ Algorithm comparison in terms of RNMSE averaged over 50 independent experiments. An enlarged version can be found in SI Section~\ref{app:simulation}.}
\label{fig:comparisionK}
\end{figure}

\section{Discussions}
There are several issues that we plan to further explore in the future. First, we are working on establishing the minimal sampling rate necessary to yield adequate numerical estimates of the first derivative matrix (see Eq.~\eqref{datamatrix}). Second, further results not shown in this paper indicate that RNMSE is high when dynamic noise and measurement noise are high. We are currently working on finer characterisation of the ``quality'' of the identification in terms of the Signal-to-Noise ratio. 


%
\onecolumn

\begin{center}
	{\Large \textbf{Supporting Information}}
\end{center}
\numberwithin{equation}{section}

\makeatletter
\renewcommand{\thesection}{S\arabic{section}}   
\renewcommand{\thesubsection}{S\arabic{subsection}}   
\renewcommand{\thetable}{S\arabic{table}}   
\renewcommand{\thefigure}{S\arabic{figure}}

\numberwithin{equation}{subsection}

\subsection{Notation}
\label{app:notation}
\paragraph*{Notation}
The notation in this paper is standard. Bold symbols are used to denote vectors and matrices.
For a matrix $\mathbf{\dic} \in \R^{M \times N}$, $\dic_{i,j} \in \R$ denotes the element in the $i^{th}$ row and $j^{th}$ column, $\dic_{i,:} \in \mathbb{R}^{1 \times N}$ denotes its $i^{th}$ row, $\dic_{:,j} \in \R^{M \times 1}$ denotes its $j^{th}$ column.
For a column vector $\boldsymbol{\alpha} \in \mathbb{R}^{N\times 1}$, $\alpha_i$
denotes its $i^{th}$ element. In particular, $\bI_L$ denotes the identity matrix of size $L\times L$. We simply use $\bI$ when the dimension is obvious from  context.
$\|\bw\|_1$ and $\|\bw\|_2$ denote the $\ell_1$ and $\ell_2$ norm of the vector $\bw$, respectively. $\|\bw\|_0$ denotes the $\ell_0$ norm of the vector $\bw$, which counts the number of nonzero elements in the vector $\bw$.  
$\diag\left[\gamma_1,\ldots, \gamma_N\right]$ denotes a diagonal matrix with principal diagonal elements being $\gamma_1, \ldots, \gamma_N$. 
$\E(\boldsymbol{\alpha})$ stands for the expectation of the stochastic variable $\boldsymbol{\alpha}$. $\propto$ means ``proportional to''.
$\blkdiag[\dic^{\left[1\right]}, \ldots, \dic^{\left[C\right]}]$ denotes a block diagonal matrix with principal diagonal blocks being $\dic^{\left[1\right]}, \ldots, \dic^{\left[C\right]}$ in turn.
$\trace(\dic)$ denotes the trace of $\dic$.
A matrix $\dic \succeq \mathbf{0}$ means $\dic$ is positive semidefinite.
A vector $\bgamma \succeq \mathbf{0}$ means each element in $\bgamma$ is nonnegative.

\subsection{Discussion on Selection of Basis Fucntions}
\label{app:basis}
Some \emph{a priori} knowledge of the field for which the models are developed can be helpful. Indeed, depending on the field for which the dynamical model needs to be built, only a few typical nonlinearities specific to this field need to be considered. For example, the class of models that arise from biochemical reaction network typically involves nonlinearities that capture fundamental biochemical kinetic laws, e.g. first-order functions $f(\left[ S\right] )=\left[ S\right]$, mass action functions $f(\left[ S_{1}\right] ,\left[ S_{2}\right] )=\left[ S_{1}\right]\cdot \left[ S_{2}\right]$,  Michaelis-Menten functions $f(\left[ S\right])=V_{\max }\left[ S\right] /(K+\left[ S\right] )$, or Hill functions $f(\left[ S\right] )=V_{\max }\left[ S\right] ^{h}/(K^{h}+[S]^h)$.

\subsection{Estimation of the First Derivative}	\label{app:deri}

Estimating time derivatives from noisy data in continuous-time systems can either be achieved using a measurement equipment with a sufficiently high sampling rate, or by using state-of-the-art mathematical approaches \cite{de2013derivative}.
Estimation of derivatives is key to the identification procedure \cite{de2013derivative}. As pointed out in \cite{papachristodoulou2007determining}, the identification problem is generally solved through discretisation of the proposed model. Assuming that samples are taken at sufficiently short time intervals, various discretisation methods can be applied. Typically, a forward Euler discretisation is used to approximate first derivatives, i.e. $\by_i$ can be defined as 
$$\by_i\define
\left[\frac{x_{i2}-x_{i1}}{\Delta t},\ldots,\frac{x_{i,{M+1}}-x_{iM}}{\Delta t}\right]^\top\in {\mathbb{R}}^{M\times 1}.$$
In this paper, the local polynomial regression framework in \cite{de2013derivative} is applied to estimate $\dot{\bx}(t)$. Forward Euler discretisation and central difference discretisation are special cases of the local polynomial regression framework.

\begin{proposition}[Proposition 1 in \cite{de2013derivative}]
Consider the bivariate data $(t_1, Y_1), \ldots, (t_M, Y_M)$.
Assume data are equispace-sampled and let $\sum_{j=1}^{k} w_j = 1$. For $k+1 \leq i \leq n-k$, the weights $w_j$ are chosen as:
\begin{equation}
\begin{aligned}
w_j = \frac{6j^2}{k(k+1)(k+2)}, \ \ \ j = 1, \ldots, k.
\label{}
\end{aligned}
\end{equation}
Based on these weights, the first derivative can be approximated as:
\begin{equation}
\begin{aligned}
\dot{Y}_i= \sum_{j=1}^k w_j \cdot \left(\frac{Y_{i+j}-Y_{i-j}}{t_{i+j}-t_{i-j}}\right).
\label{}
\end{aligned}
\end{equation}
\end{proposition}

\subsection{Derivation of Correlated Noise}
\label{app:noise}

Without loss of generality and to ease notation, we consider the scalar case.
In the scalar case, the system equation is written:
\begin{equation}
\begin{aligned}
\dot{x}_t = g(x_t)+\eta_{x_t},
\label{eq:appsys}
\end{aligned}
\end{equation}
while the measurement equation is given as:
\begin{equation}
\begin{aligned}
y_t= x_t+\epsilon_t,
\label{eq:appmea}
\end{aligned}
\end{equation}
where $\epsilon_t$ is the measurement noise, which is assumed to be Gaussian i.i.d.
We can simply use Taylor series expansion to expand $g(x_t)$ from $y_t$:
\begin{equation}
\begin{aligned}
 g(x_t) & = g(y_t-\epsilon_t) \\=& 	g(y_t)
\myunderbrace{-g'(y)|_{y = y_t}\epsilon_t 
+\mathcal{O} (\epsilon^2(t))}{\textbf{Correlated Gaussian noise}} \\
& \\
=& g(y_t) + \eta_{y_t}.
\label{eq:appext}
\end{aligned}
\end{equation}

Therefore, if we can estimate $\dot{x}$ from $y$ properly, we can write the following
\begin{equation}
\begin{aligned}
\dot{x}_{\textbf{estimate}}(t) &= g(y(t))\myunderbrace{+\eta_x(t) +\eta_y(t)}{\textbf{new noise}} \\
& \\
& = g(y(t)) + \eta(t).
\label{eq:appmodel}
\end{aligned}
\end{equation}
Clearly, $\eta(t)$ is not independent and identically distributed anymore.

\subsection{Derivation of the Cost Function}
\label{app:cost}
To derive the cost function in Section \ref{sec:costfunction}, we first introduce the posterior mean and variance
\begin{align}
	\mean&= \variance \dic^\top \Covinv \by, \label{mean2} \\
	\variance&= (\dic^\top \Covinv \dic+\bG^{-1})^{-1}. \label{variance2} 
\end{align}
Since the data likelihood $\Prob(\by|\bw)$ is Gaussian, 
\begin{equation}
\begin{aligned}
&\mathcal{N}(\by|\dic\bw, \Covinv^{-1}) \\
=& \frac{1}{\left(2\pi\right)^{M/2}|\Covinv|^{-1/2}}\exp \left[-\frac{1}{2}\left(\by-\dic\bw\right)^{\top}\Covinv \left(\by-\dic\bw\right)\right],
\end{aligned}
\end{equation}
we can write the marginal likelihood as
\begin{equation}
\begin{aligned}
& \int \mathcal{N}(\by|{\dic} {\bw},\Noise)\bN(\bw|\mathbf{0},\bG)\prod_{j=1}^{N}\prior(\hyper_{j})d\bw \\
=&\frac{1}{\left(2\pi\right)^{M/2}|\Covinv|^{-1/2}}\frac{1}{\left(2\pi\right)^{N}}
\int \exp\{-E(\bw)\}d\bw \prod_{j=1}^{N}\prior(\hyper_{j})
,
\label{integral}
\end{aligned}
\end{equation}
where
\begin{equation}
\begin{aligned}
E(\bw)=\frac{1}{2}\left(\by-\dic\bw\right)^{\top}\Covinv \left(\by-\dic\bw\right)+\frac{1}{2}\bw^\top\bG^{-1}\bw.
\end{aligned}
\end{equation}
Equivalently, we get
\begin{equation}
\begin{aligned}
E(\bw)=\frac{1}{2}(\bw-\mean)^\top\variance^{-1}(\bw-\mean)+E(\by),
\label{integral2}
\end{aligned}
\end{equation}
where $\mean$ and $\variance$ are given by (\ref{mean2}) and (\ref{variance2}).

We first show the data-dependent term is convex in $\bw$ and $\bgamma$.
From (\ref{mean2}), (\ref{variance2}), the data-dependent term can be re-expressed as
\begin{equation}
\begin{aligned}
E(\by)
=& \frac{1}{2}\left(\by^\top\Covinv\by-\by^\top\Covinv\dic\variance \dic^\top \Covinv \by\right)  \\
=&\frac{1}{2}\left(\by^\top\Covinv\by-\by^\top\Covinv\dic\variance\variance^{-1}\variance \dic^\top \Covinv \by\right) \\
=&  \frac{1}{2}\left(\by-\dic\mean\right)^{\top}\Covinv \left(\by-\dic\mean\right)+ \frac{1}{2}\mean^\top\bG^{-1}\mean \\
=&\min_{\bw}  \left[\frac{1}{2}\MSE+ \frac{1}{2}\bw^\top\bG^{-1}\bw\right].
\label{data-dependent-term}
\end{aligned}
\end{equation}

Using (\ref{integral2}), we can evaluate the integral in (\ref{integral}) to obtain
\begin{equation}
\begin{aligned}
\int \exp\{-E(\bw)\}d\bw=\exp\{-E(\by)\}(2\pi)^{N}|\variance|^{1/2}.
\end{aligned}
\end{equation}
Applying a $-2\log(\cdot)$ transformation to (\ref{integral}), we have
\begin{equation}
\begin{aligned}
&-2\log\left[\frac{1}{\left(2\pi\right)^{M/2}|\Covinv|^{-1/2}}\frac{1}{\left(2\pi\right)^{N}}
\int \exp\{-E(\bw)\}d\bw \prod_{j=1}^{N}\prior(\hyper_{j})\right] \\
= & (M+N)\log2\pi-\log |\Covinv| +\log  |\bG|+\log |\bG^{-1}+\dic^\top\Covinv\dic | \\
& +\sum_{j=1}^{N}p(\hyper_{j})+ \left( \by-\dic\bw\right)^{\top}\Covinv \left( \by-\dic\bw\right)+ \bw^\top\bG^{-1}\bw.
\label{integral3}
\end{aligned}
\end{equation}

Therefore we get the following cost function to be minimised in~\eqref{eq:cost} over $\bw, \bgamma, \Covinv$
\begin{equation}
\begin{aligned}
& \mathcal{L}(\bw, \bgamma, \Covinv) \\ 
=&- \log |\Covinv| +\log  |\bG|
+\log |\bG^{-1}+\dic^\top\Covinv \dic |\\
&+ \MSE+ \bw^\top\bG^{-1}\bw +\sum_{j=1}^{N}p(\hyper_{j}).
\notag
\end{aligned}
\end{equation}

\subsection{Some Discussion on Algorithm}
\label{s:discussion}
{\it
	\begin{remark}
		It should be noted that when noise is Gaussian i.i.d. with \underline{\emph{known}} variance, sparse Bayesian learning algorithms are provably better than classic Group Lasso algorithms in terms of mean square error \cite{aravkin2014convex}.
	\end{remark}
}

{\it
	\begin{remark}
		The initialisation step is important (line $3$ of in Algorithm~\ref{alg:summary}). In special cases where the process noise in \eqref{eq:expansion} is Gaussian i.i.d and there is no measurement noise, 
		$\Covinv$
		can be fixed to $\lambda^{-1} \bI$ for all $k$, where $\lambda$ is a positive real number, i.e. no update through \eqref{alg:rwcovariance} is carried out. In such situations, $\lambda$ can be treated as the equivalent of the regularisation/trade-off parameter in a Group Lasso algorithm~\eqref{alg:rwcovariance}. In such cases, cross validation can be implemented through variations of the initialisation values.  
	\end{remark}
}

{\it
	\begin{remark}
		When the model obtained is used for prediction purposes, the inverse covariance estimation procedure in~\eqref{alg:rwcovariance} can be used for quantification of the prediction uncertainty or risk.
	\end{remark}
}

\subsection{ADMM Implementation}
\label{s:admm}
Essentially, Algorithm~\eqref{alg:summary} consists of a reweighted Group Lasso algorithm~\eqref{alg:rwglasso} and a reweighted inverse covariance estimation algorithm~\eqref{alg:rwcovariance}. Algorithm~\eqref{alg:summary} can be implemented using the Alternating Direction Method of Multipliers (ADMM) \cite{boyd2011distributed}. This ADMM parallelisation allows to distribute the algorithmic complexity to different threads and to build a platform for scalable distributed optimisation. This is key to be able to deal with problems of large dimensions. ADMM can be used to obtain solutions to problems of the following form:
\begin{equation}
\label{admm-w-nonconvex-2}
\begin{split}
\min_{\bw}\,\,\,\,\,\, & \,f(\bw)+g(\bz),\\
\mathrm{subject}\,\,\mathrm{to}\,\,\,\,\,\,\, &
P\bw+Q\bz=\bc,
\end{split}
\end{equation}
where $\bw \in \R^n$ and $\bz \in \R^m$, $P\in \R^{p\times n}$, $Q \in \R^{p \times m}$, and $\bc \in \R^p$. The functions $f(\cdot)$ and $g(\cdot)$ are convex, but can be nonsmooth, e.g. weighted $\ell_1$ norm. The first step of the method consists in forming the augmented Lagrangian
\begin{equation}
\begin{aligned}
L_{\rho}=&f(\bw)+g(\bz)+\bu^\top(P\bw+Q\bz-c)+\\
&{\rho}/{2}\|P\bw+Q\bz-c\|_2^2.
\end{aligned}
\end{equation}
After that optimisation programmes with respect to different variables can be solved separately as follows:
\begin{equation}
\begin{aligned}
\bw^{\tau+1}&:= \argmin\limits_{\bw}\left(f(\bw)+\frac{\rho}{2}\|P\bw+Q\bz^{\tau}-\bc+\bu^{\tau}\|_2^2\right)  \notag\\
\bz^{\tau+1}&:=\argmin\limits_{\bz}\left(g(\bz)+\frac{\rho}{2}\|P\bw^{\tau+1}+Q\bz-\bc+\bu^{\tau}\|_2^2\right)  \notag\\
\bu^{\tau+1}&:= \bu^{\tau}+P\bw^{\tau+1}+Q\bz^{\tau+1}-\bc.
\end{aligned}
\end{equation}
If $g(z)$ is equal to $\lambda \|z\|_1$, then the update on $z$ is simply
\[
\bz^{\tau+1} = S_{\lambda/\rho}(P \bw^{\tau+1} + \bu^{\tau}),
\]
where $S_{\lambda/\rho}$ is the soft thresholding operator defined as
\[
S_{\lambda/\rho}(x) = 
  \begin{cases} 
   x- \lambda/\rho & \text{if } x > \lambda/\rho\\
   0       & \text{if } |x| < \lambda/\rho \\
x+\lambda/\rho & \text{if } x < -\lambda/\rho
  \end{cases}
\]
or
\[
S_{\lambda/\rho}(x)=\max(0, x- \lambda/\rho)-\max(0, -x-\lambda/\rho).\]
Based on the above, we can design a simple algorithm to solve a nonsmooth optimisation problem in a decentralised fashion. Moreover, this algorithm converges provided the following stopping criterion is satisfied:
\[
\|\bw^{\tau}-\bz^{\tau}\|_2\leq\epsilon_{primal},~~\|\rho(\bz^{\tau}-\bz^{\tau-1})\|_2 \leq \epsilon_{dual},
\]
where, the tolerances $\epsilon_{primal}>0$ and $\epsilon_{dual}>0$ can be set via an ``absolute plus relative'' criterion, e.g.
\begin{equation}
\begin{aligned}
\epsilon_{primal}&= \sqrt{n}\epsilon_{abs}+ \epsilon_{rel}\max(\|\bw^{\tau}\|_2,\|\bz^{\tau}\|_2),\notag\\
\epsilon_{dual}&=  \sqrt{n}\epsilon_{abs}+ \epsilon_{rel}\rho\|\bu^{\tau}\|,
\end{aligned}
\end{equation}
where $\epsilon_{abs}$ and $\epsilon_{rel}$ are absolute and relative tolerances. More details can be found in \cite{boyd2011distributed}.

%

More specifically, step \eqref{alg:rwcovariance} can be solved using ADMM instead:
\begin{equation} 
\label{w-nonconvex-2}
\begin{split}
\min_{\bw}\,\,\,\,\,\, & \,\MSEk+2\sum_{i=1}^{\n} \|\bz_i\|_2 ,\\
\mathrm{subject}\,\,\mathrm{to}\,\,\,\,\,\,\, &
\theta_{i}^k \bw_i - \bz_i= 0, \ \ \ i= 1, \ldots, N,
\end{split}
\end{equation}
The optimisation programmes with respect to different variables can be solved separately as follows:
\begin{equation}
\begin{aligned}
\bw^{\tau+1}&=\left(\dic^\top \Covinv^{k}\dic+\rho \bI\right)^{-1}\left(\dic^\top\Covinv^{k} \by +\rho \left(\bz^k-\bu^k\right)\right),  \notag\\
\bz^{\tau+1}_i&= \mathcal{S}_{\lambda/\rho}\left(\theta_{i}^k\bw^{\tau+1}_i+\bu^{\tau}\right), \ \ \ \ \ i =1, \ldots, N, \\
\bu^{\tau+1}&= \bu^\tau+\boldsymbol{\theta}\bw^{\tau+1}-\bz^{\tau+1}.
\end{aligned}
\end{equation}
$\mathcal{S}$ is the vector soft thresholding operator $\mathcal{S}_\kappa: \R^C \rightarrow \R^C$
\begin{equation}
\begin{aligned}
\mathcal{S}_\kappa(a) = (1-\kappa/\| a\|)_+a,
\label{}
\end{aligned}
\end{equation}
where $\mathcal{S}_\kappa(0) = 0$. This formula reduces to the scalar soft thresholding operator when $a$ is a scalar.
More details can be found in \cite{boyd2011distributed}.

\subsection{Simulation}\label{app:simulation}
In this section, we put the constructed dictionary matrix and two enlarged figures.

\begin{figure*}[!]
{\small
\begin{equation}
\begin{aligned}
\dic =\left[
\begin{array}{ccccccc}
x_{11} & \ldots & x_{81} & \text{hill}(x_{11}, 1, 0 , 3) & \ldots &
\text{hill}(x_{81}, 1, 3, 3)  & 1 \\
\vdots &  & \vdots & \vdots &  & \vdots & \vdots  \\
x_{1M}& \ldots & x_{8M} & \text{hill}(x_{1M}, 1,0,3) & \ldots &
\text{hill}(x_{8M}, 1, 3,3)  & 1 
\end{array}
\right] \in \R^{M\times 25}.
\label{oscillator-dic}
\end{aligned}
\end{equation}}
\end{figure*}

\begin{figure}
\centering
\subfigure[Root of Normalised Mean Square Error averaged over 50 independent experiments using group Lasso algorithm (first iteration of Algorithm~\ref{alg:summary})]{\label{fig:rnmse_w_iter_1}
\includegraphics[scale =0.4]{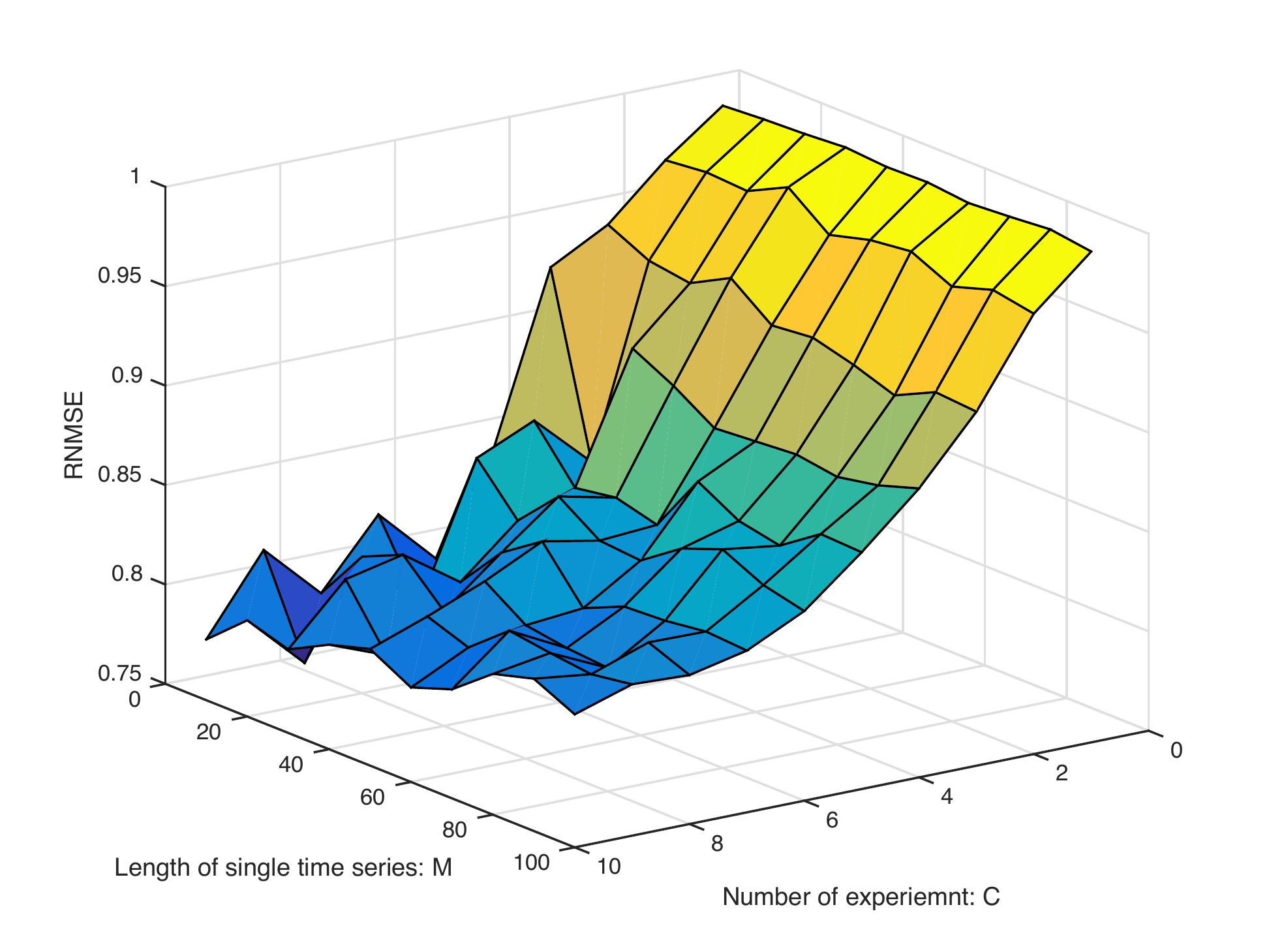}}
\subfigure[Root of Normalised Mean Square Error averaged over 50 independent experiments using Algorithm~\ref{alg:summary} with maximal iteration number $k_{\max} = 5$.
]{\label{fig:rnmse_w_iter_end}
\includegraphics[scale =0.4]{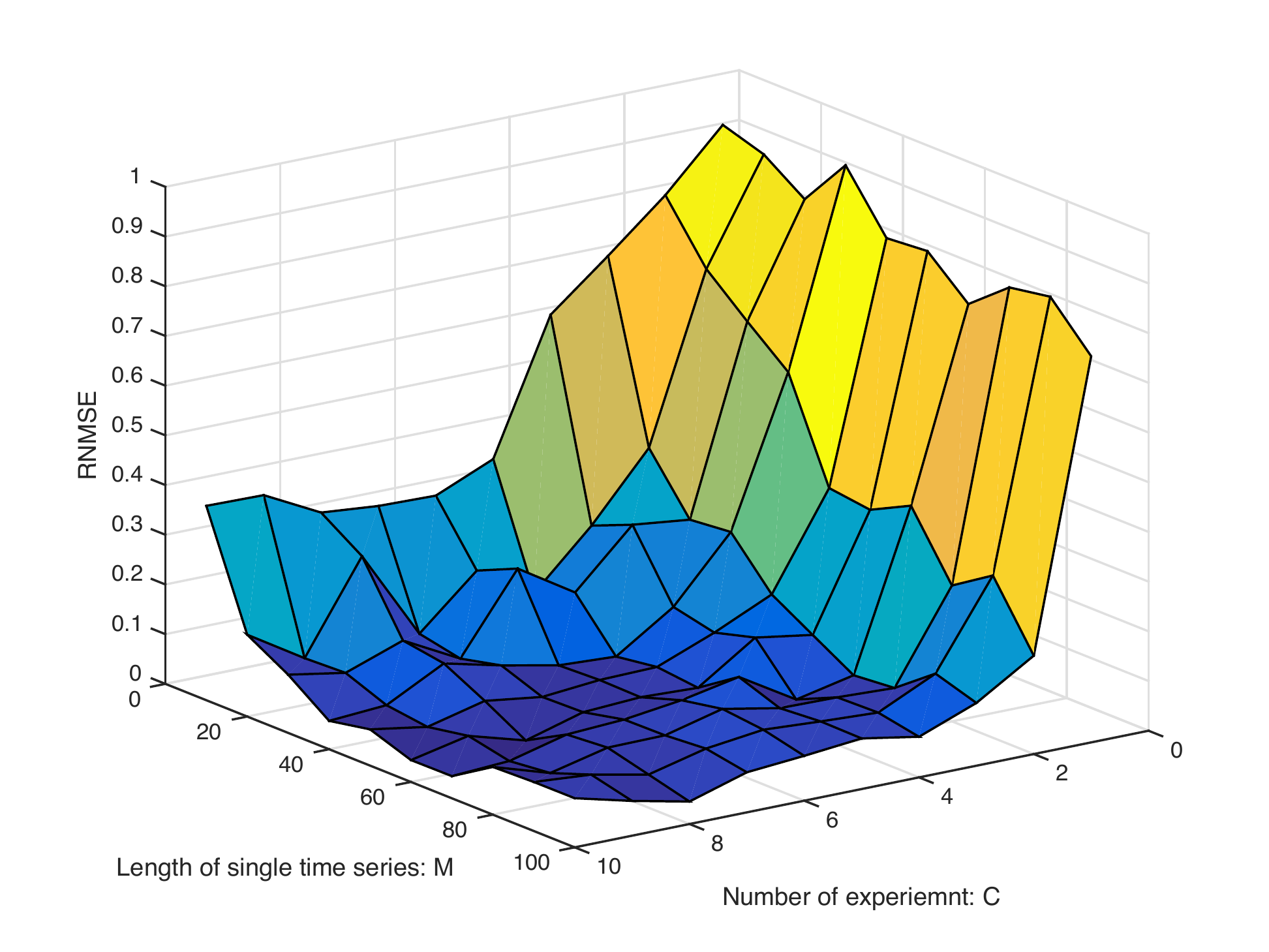}}
\caption[A set of four subfigures.]{Algorithm comparison in terms of RNMSE.}
\label{fig:SIcomparisionK}
\end{figure}


\begin{thebibliography}{10}

\bibitem{supp}
\url{http://arxiv.org/abs/1509.05153}
\bibitem{ljung1999system}
L.~Ljung, \emph{System Identification: Theory for the User}.\hskip 1em plus
  0.5em minus 0.4em\relax Prentice Hall, 1999.

\bibitem{kaltenbach2009systems}
H.-M. Kaltenbach, S.~Dimopoulos, and J.~Stelling, ``Systems analysis of
  cellular networks under uncertainty,'' \emph{FEBS letters}, vol. 583, no.~24,
  pp. 3923--3930, 2009.

\bibitem{vanlier2013parameter}
J.~Vanlier, C.~Tiemann, P.~Hilbers, and N.~van Riel, ``Parameter uncertainty in
  biochemical models described by ordinary differential equations,''
  \emph{Mathematical biosciences}, vol. 246, no.~2, pp. 305--314, 2013.

\bibitem{tipping2001sparse}
M.~Tipping, ``Sparse bayesian learning and the relevance vector machine,''
  \emph{The Journal of Machine Learning Research}, vol.~1, pp. 211--244, 2001.

\bibitem{wipf2011latent}
D.~Wipf, B.~Rao, and S.~Nagarajan, ``Latent variable bayesian models for
  promoting sparsity,'' \emph{Information Theory, IEEE Transactions on},
  vol.~57, no.~9, pp. 6236--6255, 2011.

\bibitem{ingolia2008systems}
N.~T. Ingolia and J.~S. Weissman, ``Systems biology: reverse engineering the
  cell,'' \emph{Nature}, vol. 454, no. 7208, pp. 1059--1062, 2008.

\bibitem{boyd1987linear}
S.~Boyd, L.~El~Ghaoul, E.~Feron, and V.~Balakrishnan, \emph{Linear matrix
  inequalities in system and control theory}.\hskip 1em plus 0.5em minus
  0.4em\relax Society for Industrial Mathematics, 1987, vol.~15.

\bibitem{horn1990matrix}
R.~Horn and C.~Johnson, \emph{Matrix analysis}.\hskip 1em plus 0.5em minus
  0.4em\relax Cambridge university press, 1990.

\bibitem{cerone2011enforcing}
V.~Cerone, D.~Piga, and D.~Regruto, ``Enforcing stability constraints in
  set-membership identification of linear dynamic systems,'' \emph{Automatica},
  vol.~47, no.~11, pp. 2488--2494, 2011.

\bibitem{zavlanos2011inferring}
M.~Zavlanos, A.~Julius, S.~Boyd, and G.~Pappas, ``Inferring stable genetic
  networks from steady-state data,'' \emph{Automatica}, vol.~47, no.~6, pp.
  1113--1122, 2011.

\bibitem{boyd1994lmi}
S.~Boyd, L.~El~Ghaoui, E.~Feron, and V.~Balakrishnan, \emph{{Linear Matrix
  Inequalities in System and Control Theory}}.\hskip 1em plus 0.5em minus
  0.4em\relax Society for Industrial Mathematics, 1994.

\bibitem{bishop2006pattern}
C.~Bishop, \emph{Pattern Recognition and Machine Learning}.\hskip 1em plus
  0.5em minus 0.4em\relax Springer New York, 2006, vol.~4.

\bibitem{palmer2006variational}
J.~Palmer, D.~Wipf, K.~Kreutz-Delgado, and B.~Rao, ``Variational {EM}
  algorithms for non-{Gaussian} latent variable models,'' \emph{Advances in
  neural information processing systems}, vol.~18, p. 1059, 2006.

\bibitem{sriperumbudur2009convergence}
B.~K. Sriperumbudur and G.~R. Lanckriet, ``On the convergence of the
  concave-convex procedure.'' in \emph{NIPS}, vol.~9, 2009, pp. 1759--1767.

\bibitem{aravkin2014convex}
A.~Aravkin, J.~V. Burke, A.~Chiuso, and G.~Pillonetto, ``Convex vs non-convex
  estimators for regression and sparse estimation: the mean squared error
  properties of ard and glasso,'' \emph{The Journal of Machine Learning
  Research}, vol.~15, no.~1, pp. 217--252, 2014.

\bibitem{boyd2004convex}
S.~Boyd and L.~Vandenberghe, \emph{Convex optimisation}.\hskip 1em plus 0.5em
  minus 0.4em\relax Cambridge university press, 2004.

\bibitem{TianshiTAC}
T.~Chen, M.~Andersen, L.~Ljung, A.~Chiuso, and G.~Pillonetto, ``System
  identification via sparse multiple kernel-based regularization using
  sequential convex optimization techniques,'' \emph{Automatic Control, IEEE
  Transactions on}, vol.~59, no.~11, pp. 2933--2945, 2014.

\bibitem{boyd2011distributed}
S.~Boyd, N.~Parikh, E.~Chu, B.~Peleato, and J.~Eckstein, ``Distributed
  optimization and statistical learning via the alternating direction method of
  multipliers,'' \emph{Foundations and Trends in Machine Learning}, vol.~3,
  no.~1, pp. 1--122, 2011.

\bibitem{strelkowa2010switchable}
N.~Strelkowa and M.~Barahona, ``Switchable genetic oscillator operating in
  quasi-stable mode,'' \emph{Journal of The Royal Society Interface}, p.
  rsif20090487, 2010.

\bibitem{de2013derivative}
K.~De~Brabanter, J.~De~Brabanter, B.~De~Moor, and I.~Gijbels, ``Derivative
  estimation with local polynomial fitting,'' \emph{The Journal of Machine
  Learning Research}, vol.~14, no.~1, pp. 281--301, 2013.


\bibitem{claire2015cdc}
Y. Chang, R. Dobbe, P. Bhushan, J. W. Gray, C. J. Tomlin, ``Retrieving common dynamics of gene regulatory networks under various perturbations,'' in \emph{Proceeding of Conference on Decision and Control}, 2015.



\end{thebibliography}

\begin{thebibliography}{10}

\bibitem{papachristodoulou2007determining}
A.~Papachristodoulou and B.~Recht, ``Determining interconnections in chemical
  reaction networks,'' in \emph{American Control Conference, 2007.
  ACC'07}.\hskip 1em plus 0.5em minus 0.4em\relax IEEE, 2007, pp. 4872--4877.


\end{thebibliography}
\end{document}